\documentclass[10pt,twocolumn,twoside,aps,prl,showpacs,superscriptaddress,a4paper]{revtex4-1}
\usepackage{graphicx}
\usepackage{amsmath,amssymb}
\usepackage{color}
\usepackage{bbm}
\usepackage{hyperref}
\usepackage{theorem}
\usepackage{latexsym}

\newtheorem{theorem}{Theorem}

\newtheorem{corollary}[theorem]{Corollary}

\newtheorem{definition}[theorem]{Definition}

\newtheorem{lemma}[theorem]{Lemma}

\newlength{\blank}
\newenvironment{proof}[1][{\hspace{-\blank}}]{{\noindent\textbf{Proof~{#1}.\ }}}{\hfill\qed}
\newenvironment{proofthm}[1]{{\noindent\textbf{Proof~{#1}\ }}}{\hfill\qed}
\newenvironment{remark}{\noindent\textbf{Remark\ }}{}
\newcommand{\ket}[1]{|#1\rangle}
\newcommand{\bra}[1]{\langle#1|}

\newcommand{\CNOT}{\operatorname{CNOT}}

\mathchardef\ordinarycolon\mathcode`\:
\mathcode`\:=\string"8000
\def\vcentcolon{\mathrel{\mathop\ordinarycolon}}
\begingroup \catcode`\:=\active
  \lowercase{\endgroup
  \let :\vcentcolon
  }

\newcommand{\nc}{\newcommand}
\nc{\rnc}{\renewcommand}
\nc{\beq}{\begin{equation}}
\nc{\eeq}{{\end{equation}}}
\nc{\beqa}{\begin{eqnarray}}
\nc{\eeqa}{\end{eqnarray}}
\nc{\lbar}[1]{\overline{#1}}
\nc{\ketbra}[2]{|#1\rangle\!\langle#2|}
\nc{\proj}[1]{| #1\rangle\!\langle #1 |}
\nc{\avg}[1]{\langle#1\rangle}
\nc{\Rank}{\operatorname{Rank}}
\nc{\smfrac}[2]{\mbox{$\frac{#1}{#2}$}}
\nc{\tr}{\operatorname{Tr}}
\nc{\ox}{\otimes}
\nc{\dg}{\dagger}
\nc{\dn}{\downarrow}
\nc{\cA}{\mathcal{A}}
\nc{\cB}{\mathcal{B}}
\nc{\cC}{\mathcal{C}}
\nc{\cD}{\mathcal{D}}
\nc{\cE}{\mathcal{E}}
\nc{\cF}{\mathcal{F}}
\nc{\cG}{\mathcal{G}}
\nc{\cH}{\mathcal{H}}
\nc{\cI}{\mathcal{I}}
\nc{\cJ}{\mathcal{J}}
\nc{\cK}{\mathcal{K}}
\nc{\cL}{\mathcal{L}}
\nc{\cM}{\mathcal{M}}
\nc{\cN}{\mathcal{N}}
\nc{\cO}{\mathcal{O}}
\nc{\cP}{\mathcal{P}}
\nc{\cR}{\mathcal{R}}
\nc{\cS}{\mathcal{S}}
\nc{\cT}{\mathcal{T}}
\nc{\cX}{\mathcal{X}}
\nc{\cZ}{\mathcal{Z}}
\nc{\csupp}{{\operatorname{csupp}}}
\nc{\qsupp}{{\operatorname{qsupp}}}
\nc{\var}{\operatorname{var}}
\nc{\rar}{\rightarrow}
\nc{\lrar}{\longrightarrow}
\nc{\polylog}{\operatorname{polylog}}
\nc{\1}{{\openone}}
\nc{\id}{{\operatorname{id}}}

\nc{\RR}{{{\mathbb R}}}
\nc{\CC}{{{\mathbb C}}}
\nc{\FF}{{{\mathbb F}}}
\nc{\NN}{{{\mathbb N}}}
\nc{\ZZ}{{{\mathbb Z}}}
\nc{\PP}{{{\mathbb P}}}
\nc{\QQ}{{{\mathbb Q}}}
\nc{\UU}{{{\mathbb U}}}
\nc{\EE}{{{\mathbb E}}}

\nc{\qed}{{\hfill$\Box$}}

\def\>{\rangle}
\def\<{\langle}

\begin{document}

\title{Operational Resource Theory of Coherence}

\author{Andreas Winter}
\email{andreas.winter@uab.cat}
\affiliation{ICREA -- Instituci\'{o} Catalana de Recerca i Estudis Avan\c{c}ats, %
Pg.~Llu\'{\i}s Companys, 23, 08010 Barcelona, Spain}
\affiliation{F\'{\i}sica Te\`{o}rica: Informaci\'{o} i Fen\`{o}mens Qu\`{a}ntics, %
Universitat Aut\`{o}noma de Barcelona, ES-08193 Bellaterra (Barcelona), Spain}

\author{Dong Yang}
\email{dyang@cjlu.edu.cn}
\affiliation{F\'{\i}sica Te\`{o}rica: Informaci\'{o} i Fen\`{o}mens Qu\`{a}ntics, %
Universitat Aut\`{o}noma de Barcelona, ES-08193 Bellaterra (Barcelona), Spain}
\affiliation{Laboratory for Quantum Information, China Jiliang University, Hangzhou, Zhejiang 310018, China}

\date{16 January 2016}

\begin{abstract}
We establish an operational theory of coherence (or of superposition) in
quantum systems, by focusing on the optimal rate of performance of certain tasks.
Namely, we introduce the two basic concepts --- ``coherence distillation''
and ``coherence cost'' in the processing quantum states under so-called
\emph{incoherent operations} [Baumgratz/Cramer/Plenio, Phys. Rev. Lett. 113:140401 (2014)].
We then show that in the asymptotic limit of many copies of a state,
both are given by simple single-letter formulas: the distillable coherence 
is given by the relative entropy of coherence (in other words, we give the 
relative entropy of coherence its operational interpretation), and the coherence 
cost by the coherence of formation, which is an optimization over convex
decompositions of the state.
An immediate corollary is that there exists no bound coherent state in the sense 
that one would need to consume coherence to create the state but no coherence could 
be distilled from it. Further we demonstrate that the coherence theory is 
generically an irreversible theory by a simple criterion that completely characterizes 
all reversible states. 
\end{abstract}


\maketitle

{\it Introduction.}---
The universality of the superposition principle is the
fundamental non-classical characteristic of quantum mechanics: 
Given any configuration space $X$, its elements $x$ label an
orthogonal basis $\ket{x}$ of a Hilbert space, and we have
all superpositions $\sum_x \psi_x \ket{x}$ as the possible 
states of the system. In particular, we could choose a completely 
different orthonormal basis as an equally valid computational
basis, in which to express the superpositions. 
However, often a basis is distinguished, be it the eigenbasis
of an observable or of the system's Hamiltonian, so that conservation
laws or even superselection rules apply.
In such a case, the eigenstates $\ket{x}$ are distinguished as ``simple''
and superpositions are ``complex''. Indeed, in the presence of
conservation laws, structured superpositions of eigenstates can serve
as so-called ``reference frames'', which are resources to overcome
the conservation laws~\cite{AharonovSusskind,Kitaev,BRS-ref}.
Based on this idea, {\AA}berg~\cite{Aaberg:superposition} and more 
recently Baumgratz \emph{et al.}~\cite{BCP14} have proposed to consider
any non-trivial superposition as a resource, and to create a
theory in which computational basis states and their probabilistic
mixtures are for free (or worthless), and operations preserving
these ``incoherent'' states are free as well. This suggests that 
coherence theory can be regarded as a resource theory. 

Let us briefly recall the general structure of a quantum resource theory 
(QRT) and basic questions that be asked in a QRT through entanglement 
theory (ET), which is a well-known QRT. For our purposes, a QRT has 
three ingredients: (1) free states (separable state in ET), (2) 
resource states (entangled state in ET), (3) the restricted or free operations 
(LOCC in ET). A necessary requirement for a consistent QRT is that no 
resource state can be created from any free state under any free operation. 
The QRT is then the study of interconversion between resource states under 
the restricted operations. The pure resource states play a special role 
and are much more preferable because usually they are used to circumvent 
the restriction on operations (Bell states are used in teleportation to 
overcome LOCC operations, for instance). So the conversions between 
pure resource states and mixed ones are a major focus of QRTs.  
A standard unit resource measure can be constructed if the conversions 
between pure states are asymptotically reversible (entropy of entanglement 
of a pure entangled state in ET, with a Bell state as the unit). 
Then there are two basic transformation processes that are well-motivated: 
one is the so-called resource distillation, that is the transformation 
from a mixed resource state to the unit resource, and the other is 
resource formation, that is the reverse transformation from the unit 
resource to a mixed state (entanglement distillation and entanglement 
formation in ET). Because of the reversibility in the pure state conversion 
we need not worry about what kind of pure state is the target, as they 
are equivalent up to the transformation rate between them. Thus, two 
well-motivated quantities arise from the two basic processes, distillable 
resource and resource cost, which have a clear operational interpretation 
(distillable entanglement and entanglement cost in ET). The principal 
objective of the theory is the characterization of these two quantities. 
This is often a highly complex problem, but resource monotones yield 
various limits on possible transformations and achievable rates.  
Another basic question in any QRT is to ask whether the theory is 
irreversible or not. If the conversion between pure states is reversible, 
then the reversibility problem is reduced to the question whether or 
not the optimal conversion rate in the formation process is equal to 
that in the distillation process. 
If a QRT is reversible, a unique resource measure exits, quantifying the 
conversion rate between different states, so that everything about possible 
resource transformations becomes clear and simple. However, if a QRT is 
irreversible, the phenomena are ample and further interesting questions can 
be asked, for example whether there exist so-called bound resource states 
as an analogue of bound entangled states \cite{HHH-bound,Yang2005} in ET, in the 
sense that from them no resource can be distilled, but for which, in order 
to create them, nonzero resource is required. Several QRTs were constructed 
along these lines, some them indeed irreversible: entanglement theory (with 
respect to LOCC) \cite{HHHH,PV}, thermodynamics (w.r.t.~energy-conserving 
operations and thermal states) \cite{HO,Brandao-etal}, reference frames 
(w.r.t.~group-covariant operations) \cite{frame}, etc.

In this Letter, we establish an operational coherence theory in the 
framework proposed in \cite{BCP14} and \cite{Aaberg:superposition}. Namely,
first we show that the conversion between the pure coherent states is 
asymptotically reversible, so that the standard unit coherence measure 
exists. Then we introduce the basic transformation processes -- coherence 
distillation and coherence formation --, from which two basic coherence 
measures naturally arise: distillable coherence and coherence cost, 
with operational interpretations. Remarkably, both are given by information 
theoretic single-letter formulas that hence make these quantities computable. 
These results in turn allow us to formulate a simple criterion to decide whether 
a given state is irreversible or not and to show that there is no bound 
coherence. Although the main results are in the asymptotic setting, we 
also get the single copy conversion of pure states along the way. In the 
following, we state and discuss our results carefully, while all proofs 
are found in Appendix B.

\medskip
{\it Coherence as a resource theory.}---
We follow the framework of coherence theory by Baumgratz \emph{et al.}~\cite{BCP14}. 
Let $\{\ket{i}\}$ be a fixed basis in the finite dimensional Hilbert state. 
The free states called \emph{incoherent states} are those whose density matrices are 
diagonal in the basis, being of the form $\sum_i p_i\proj{i}$ where $p_i$ 
is a probability distribution and the set of incoherent states is denoted as $\Delta$. 
The resource states called \emph{coherent states} are those not of this form. 
Quantum operations are specified by a set of Kraus operators $\{K_\ell\}$ 
satisfying $\sum_\ell K_\ell^{\dagger}K_\ell = \1$;
a quantum operation can have many different Kraus representations.
The free operations, called \emph{incoherent operations (IC)} are those for which
there exists a Kraus representation $\{K_\ell\}$ such that 
$\frac{1}{\tr\rho K_\ell^\dagger K_\ell} K_\ell\rho K_\ell^{\dagger}\in\Delta$ 
for all $\ell$ and all $\rho\in \Delta$.  
Such restriction guarantees that even if one has access to individual measurement
outcomes $\ell$ of the instrument $\{K_\ell\}$,  one cannot generate coherent states
from an incoherent state. 
Under this restriction, each Kraus operator is easily seen to be of the form 
$K_\ell=\sum_i c(i)\ketbra{j(i)}{i}$ where $j(i)$ is a function from the index 
set of the computational basis, and $c(i)$ are coefficients; 
we call such Kraus operators incoherent, too. If not only $K_\ell$ but
also $K_\ell^\dagger$ is incoherent, we call it \emph{strictly incoherent},
and the corresponding quantum operation a strictly incoherent operation.
Strict incoherence can be characterized by the function $j(i)$ above
being one-to-one.
Another equivalent form of a general incoherent Kraus operator is
$K = \sum_j \ketbra{j}{\gamma_j}$, 
with $\ket{\gamma_j} \in \operatorname{span}\{\ket{i}:i\in S_j\}$
for a partition $[d] = {\dot{\bigcup}}_j S_j$. 

We will be mainly concerned with state transformations, for which
we introduce some notation: 
$\rho \stackrel{{\rm IC}}{\longmapsto} \sigma$ means that there exists an incoherent
operation $T$ such that $\sigma = T(\rho) = \sum_\ell K_\ell \rho K_\ell^\dagger$;
if the operation is strictly incoherent, we write
$\rho \stackrel{{\rm IC}_0}{\longmapsto} \sigma$. If the transformation is
obtained probabilistically, i.e.~if $\sigma \propto K_\ell\rho K_\ell^\dagger$
and $\tr K_\ell\rho K_\ell^\dagger \neq 0$ for some $\ell$, we write
$\rho \stackrel{{\rm pIC}}{\longmapsto} \sigma$ and
$\rho \stackrel{{\rm pIC}_0}{\longmapsto} \sigma$, for
probabilistic incoherent and probabilistic strictly incoherent 
transformations, respectively.

When we consider composite systems, we simply declare as incoherent the
tensor product basis of the local computational bases;
the incoherent operations are then defined with respect to 
the tensor product basis. Notice that there are several special 
transformations that are incoherent: phase and permutation unitaries. 
In particular, in a composite system of two qudits, 
$\CNOT:\ket{i}\ket{j} \mapsto \ket{i}\ket{i+j\mod d}$ is
an incoherent operation, as it is simply a permutation
of the tensor product basis vectors~\cite{BCA,Adesso-et-al}. 
In Appendix C we discuss a slightly more general and flexible model.

We define the decohering operation $\Delta(\rho)=\sum_i \bra{i} \rho \ket{i} \proj{i}$,
i.e.~the diagonal part of the density matrix. This makes the incoherent
states $\Delta$ the image of the map $\Delta$, thus justifying the slight
abuse of notation.

\medskip
{\it Pure state transformations.}---
We start by developing the theory of pure states transformations;
the main result in the present context is that in the asymptotic setting of many 
copies this becomes reversible, the rates governed by the entropy of the decohered 
state, a quantity we dub \emph{entropy of coherence}. First however, we review the
situation for exact single-copy transformations.

We start with a simple observation on ranks: Namely, let
$\varphi$ be transformed to $\psi$ by an incoherent, or more generally
a probabilistic incoherent operation, $\ket{\psi} \propto K \ket{\varphi} \neq 0$,
for an incoherent Kraus operator $K = \sum_i c_i \ketbra{j(i)}{i}$.
The rank $r$ of $\Delta(\varphi)$ is precisely the number of nonzero
diagonal entries of $\Delta(\varphi)$, which is the number of
nonzero terms in $\ket{\varphi} = \sum_{i\in R} \varphi_i \ket{i}$,
$|R|=r$. Thus,
\[
  \ket{\psi} \propto K \ket{\varphi} = \sum_{i\in R} \varphi_i c_i \ket{j(i)}
\]
has at most $r=|R|$ terms. 
This proves the following.

\begin{lemma}
  \label{lemma:rank-nonincrease}
  If $\varphi \stackrel{{\rm pIC}}{\mapsto} \psi$, then
  $\operatorname{rank} \Delta(\psi) \leq \operatorname{rank} \Delta(\varphi)$,
  i.e.~the rank of the diagonal part of pure states 
  cannot increase under incoherent operations.
\end{lemma}

\begin{theorem}[{Cf.~Du/Bai/Guo~\cite{DuBaiGuo}}]
  \label{thm:majorization}
  For two pure states $\psi = \proj{\psi}$ and $\varphi = \proj{\varphi}$,
  if $\Delta(\psi)$ majorizes $\Delta(\varphi)$, 
  $\Delta(\psi) \succ \Delta(\varphi)$, then there is an incoherent
  (in fact, a strictly incoherent) operation transforming $\varphi$ to $\psi$:
  $\varphi \stackrel{{\rm IC}_0}{\mapsto} \psi$.
  
  Conversely, if $\varphi \stackrel{{\rm IC}_0}{\mapsto} \psi$, or if
  $\varphi \stackrel{{\rm IC}}{\mapsto} \psi$ and in addition
  $\operatorname{rank} \Delta(\varphi) = \operatorname{rank} \Delta(\psi)$,
  then $\Delta(\psi) \succ \Delta(\varphi)$.
\end{theorem}

Here, the majorization relation for matrices, $\rho \succ \sigma$, 
means that the spectra 
$\operatorname{spec}(\rho) = \vec{p} = (p_1 \geq \ldots \geq p_d)$
and 
$\operatorname{spec}(\sigma) = \vec{q} = (q_1 \geq \ldots \geq q_d)$ are
in majorization order~\cite{MarshallOlkin,AlbertiUhlmann,Nielsen:majorization}:
\begin{align*}
  \forall t < d\ \sum_{i=1}^t p_i &\geq \sum_{i=1}^t q_i, \text{ and }
                 \sum_{i=1}^d p_i  =    \sum_{i=1}^d q_i.
\end{align*}

As a consequence of Theorem~\ref{thm:majorization}, just as for pure state
entanglement~\cite{Nielsen:majorization}, there is catalysis for pure state incoherent
transformations, cf.~\cite{Jonathan-catalyst}; in~\cite{Duan},
examples such that initial and final states have equal rank 
are given. An immediate corollary of Theorem~\ref{thm:majorization} is the following.

\begin{corollary}[{Baumgratz \emph{et al.}~\cite{BCP14}}]
  \label{cor:max-coherent-gen-all}
  Let $\rho$ be an arbitrary state in $\CC^d$, and
  $\Phi_d = \proj{\Phi_d} = \frac1d \sum_{ij=0}^{d-1} \ketbra{i}{j}$. 
  Then, there is an incoherent operation transforming
  $\Phi_d$ to $\rho$. 
\end{corollary}

This motivates the name \emph{maximally coherent state} for $\Phi_d$.
In addition to enabling the creation of arbitrary $d$-dimensional
coherent superpositions by incoherent means, $\Phi_d$ also allows
the implementation of arbitrary unitaries $U\in{\rm SU}(d)$~\cite{BCP14}.
Then, fixing the qubit maximally coherent pure state 
$\Phi_2 = \frac12 \sum_{ij=0}^1 \ketbra{i}{j}$ as a unit reference, we are 
ready to consider asymptotic pure states transformations with vanishing error
as the number of copies goes to infinity. Special cases of this are 
coherence concentration, the transformation from a non-maximally coherent pure 
state to the unit coherent state, and coherence dilution from unit 
coherent state to the non-maximally coherent pure one. 
As in information theory, entanglement theory and other similar cases
(cf.~\cite{CoverThomas,Wilde-book}), this simplifies the picture dramatically.
To express our result, we introduce the entropy of coherence for pure states as
\[
  C(\psi) = S\bigl(\Delta(\psi)\bigr).
\]
Here $S(\rho)=-\tr\rho\log\rho$ is the von Neumann entropy, where 
logarithms are to base 2. Note the maximum value of this functional 
among $d$-dimensional states is attained on $\Phi_d$:
$C(\Phi_d) = \log d$, in particular $C(\Phi_2) = 1$ for the
unit coherence resource.

\begin{theorem}[{Yuan/Zhou/Cao/Ma~\cite{Yuan}}]
  \label{thm:asymptotic}
  For two pure states $\psi$ and $\varphi$ and a rate $R\geq 0$,
  the asymptotic incoherent transformation
  \[
    \psi^{\ox n} \stackrel{\text{IC}}{\longmapsto} \,\,\stackrel{1-\epsilon}\approx \varphi^{\ox nR}
    \text{ as } n\rightarrow\infty,\ \epsilon\rightarrow 0,
  \]
  is possible if $R < \frac{C(\psi)}{C(\varphi)}$ and impossible
  if $R > \frac{C(\psi)}{C(\varphi)}$. 
  
  In particular, $\psi$ can be asymptotically reversibly transformed
  into $\Phi_2$, and vice versa, at optimal rate $C(\psi)$.
\end{theorem}

Here, $\rho\stackrel{1-\epsilon}{\approx}\sigma$ signifies 
that the two states have high fidelity:
$F(\rho,\sigma)\geq 1-\epsilon$ (see Appendix A for details).

Now we are ready to introduce two fundamental tasks for arbitrary mixed states,
namely asymptotic distillation of $\rho^{\ox n}$ to $\Phi_2^{\ox nR}$
and the reverse process of formation $\rho^{\ox n}$ from $\Phi_2^{\ox nR}$.
Note in this respect the fundamental importance of Theorem~\ref{thm:asymptotic},
which shows that we could equivalently elect any pure state $\psi$ 
(that is coherent) as unit reference for formation and distillation,
and all rates would change by the same factor $\frac{1}{C(\psi)}$. 
It turns out that both quantities have single-letter, additive expressions: 
the former is given by the \emph{relative entropy of coherence}, the 
latter by the \emph{coherence of formation}; both are additive.
This is in marked contrast to other resource theories, perhaps
most prominently that of entanglement under LOCC, in which the
basic operational tasks are only characterized by regularized
formulas, and the fundamental quantities, such as 
entanglement of formation~\cite{BDSW}, 
relative entropy of entanglement~\cite{relent}, etc,
are not additive~\cite{VollbrechtWerner,Hastings}.

\medskip
{\it Distillable coherence.}---
The distillation process is the process that extracts pure coherence
from a mixed state by incoherent operations. The distillable coherence 
of a state is the maximal rate at which $\Phi_2$ can be obtained from the 
given state. The precise definition is the following.

\begin{definition}
  The distillable coherence of a state $\rho$ is defined as
  \[
    C_d(\rho)=\sup R, \text{ s.t. }
    \rho^{\ox n} \stackrel{\text{IC}}{\longmapsto} \,\,\stackrel{1-\epsilon}\approx \Phi_2^{\ox nR}
    \text{ as } n\rightarrow\infty,\ \epsilon\rightarrow 0.
  \]  
\end{definition}

We look at the asymptotic setting to get rid of the error $\epsilon$,
motivated by general information theoretic practice \cite{Wilde-book,CoverThomas}
and the success of this point of view in the pure state case (Theorem~\ref{thm:asymptotic}).
By definition, $C_d$ naturally has an operational meaning as the optimal
rate performance at a natural task.
Theorem \ref{thm:distillable} shows that the distillable coherence is given 
by a closed form expression. 

\begin{theorem}
  \label{thm:distillable}
  For any state $\rho$, the distillable coherence is given by the
  \emph{relative entropy of coherence}:
  \[
    C_d(\rho) = C_r(\rho) = \min_{\sigma\in \Delta} S(\rho\|\sigma).
  \]
\end{theorem}

The relative entropy of coherence is introduced 
and studied in detail in \cite{Aaberg:superposition,BCP14}.
Here, $S(\rho\|\sigma)=\tr\rho(\log\rho-\log\sigma)$ is the quantum 
relative entropy. In \cite{Aaberg:superposition,BCP14}, it is shown that
\begin{equation}
\label{eq:C_r}
  C_r(\rho) = S\bigl(\Delta(\rho)\bigr) - S(\rho).
\end{equation}
Furthermore, that the relative entropy of coherence 
is convex in the state, and a coherence monotone, meaning
that for every incoherent transformation $T(\rho) = \sum_\ell K_\ell\rho K_\ell^\dagger$,
$C_r(\rho) \geq C_r\bigl( T(\rho) \bigr)$. In fact, it is even 
\emph{strongly monotonic}~\cite{BCP14}:
\[
  C_r(\rho) \geq \sum_\ell p_\ell C_r(\rho_\ell),
  \text{ where } p_\ell\rho_\ell = K_\ell\rho K_\ell^\dagger.
\]

But it is only due to Theorem \ref{thm:distillable} that we
can give it a clear operational interpretation in the distillation process. 
Note that in concurrent independent work, Singh \emph{et al.}~\cite{Pati-et-al}
show that the same quantity, $C_r(\rho)$, by virtue of it being equal
to the entropy difference $S\bigl(\Delta(\rho)\bigr) - S(\rho)$, 
arises also as the minimum amount of incoherent noise that has to be applied
to the state to decohere it.

\medskip
{\it Coherence cost.}---
The formation process is that which prepares a mixed state by 
consuming pure coherent states under incoherent operations. 
The coherence cost is the minimal rate at which  
$\Phi_2$ has to be consumed for preparing the given state. 

\begin{definition}
  The coherence cost of a state $\rho$ is defined as 
  \[
  C_c(\rho) = \inf R, \text{ s.t. }
                  \Phi_2^{\ox nR} \stackrel{\text{IC}}{\longmapsto} \,\,\stackrel{1-\epsilon}\approx \rho^{\ox n}
              \text{ as } n\rightarrow\infty,\ \epsilon\rightarrow 0.
  \]
\end{definition}

The next result shows that the coherence cost has a single-letter 
formula, involving a simple entropy optimization.

\begin{theorem}
  \label{thm:cost}
  For any state $\rho$, the coherence cost is given by the
  coherence of formation,
  \[
    C_c(\rho) = C_f(\rho),
  \]
  where the coherence of formation is given by 
  \[
    C_f(\rho) = \min \sum_i p_i S(\Delta(\psi_i)) \text{ s.t. } \rho = \sum_i p_i \ketbra{\psi_i}{\psi_i}.
  \]
\end{theorem}

The coherence of formation is introduced as a monotone in \cite{Aaberg:superposition}, 
where its convexity and monotonicity under incoherent operations were
observed (see Lemma~\ref{lemma:CoF-monotone} in Appendix A),
however its additivity was not remarked. A priori, one might have expected the cost to
be given by the regularization of the coherence of formation, which would have
involved infinitely many optimization problems so that the evaluation of the 
operational cost had become infeasible.
Indeed it is additivity that makes the single-letter formulas 
in Theorem \ref{thm:distillable} and Theorem \ref{thm:cost} available,
and because of its importance we record it as a theorem.

\begin{theorem}
  \label{additivity}
  Both $C_f$ and $C_r$ are additive under tensor products,
  \begin{align*}
    C_f(\rho\ox\sigma) &= C_f(\rho) + C_f(\sigma), \\
    C_r(\rho\ox\sigma) &= C_r(\rho) + C_r(\sigma).
  \end{align*}
\end{theorem}  

Theorem \ref{additivity} means that also the distillable coherence and the
coherence cost are additive, and so there are no super-additivity
or activation phenomena in the resource theory of coherence, unlike
the theory of entanglement~\cite{VollbrechtWerner,Hastings}, that of communication via
channels~\cite{Hastings,SmithYard}, and many other resource theories,
where the yield (cost) of two resources together may be strictly larger 
(smaller) than the sum of the yields (costs) of the resources
processed individually.

Based on the formulas for the distillable coherence and the coherence cost, 
we are now ready to characterize precisely the (ir-)reversible states.

\medskip
{\it (Ir-)reversibility.}--- From the definitions of $C_d$ and $C_c$, it 
is immediate that $C_d(\rho)\le C_c(\rho)$. A state is reversible if the 
equality holds, otherwise it is called irreversible. 
From Theorem~\ref{thm:asymptotic}, we see that 
all pure states are asymptotically reversible, the optimal transformation rate given by the
ratio of their entropies of coherence. In the mixed state case, 
Theorem \ref{thm:criterion} provides a simple criterion to decide whether 
the given state is reversible or not that  completely characterizes 
all the reversible states. From this, we conclude that the mixed states 
are generically irreversible and the coherence theory is an irreversible 
resource theory. However in contrast to the irreversibility
in entanglement theory~\cite{HHH-bound,VC-irr,Yang2005}, there is no 
``bound coherence'' (as an analogue of ``bound entanglement''~\cite{HHH-bound,Yang2005}), 
from which no coherence could be distilled, but for which, in order to create it, 
nonzero coherence would be required. 

\begin{theorem}
\label{thm:criterion}
A mixed state is reversible if and only if its eigenvectors are supported on the 
orthogonal subspaces spanned by a partition of the incoherent basis. 
That is,
\[
  \rho = \bigoplus_j p_j \proj{\phi_j},
\]
and each $\ket{\phi_j} \in{\cal H}_j =\operatorname{span}\{\ket{i}: i\in S_j\}$,
such that $S_j\cap S_k=\emptyset$ for all $j\neq k$ and
$\cH = \bigoplus_j \cH_j$. In other words, $\rho$ is reversible if and only if
there exists an incoherent projective measurement, consisting in fact
of the projectors $P_j$ onto ${\cal H}_j$, s.t.~$\rho = \sum_j P_j \rho P_j$,
and $P_j \rho P_j$ is proportional to a pure state.
\end{theorem}

\medskip
Note that the criterion is easy to check, and contrast this with the entanglement 
irreversibility of two-qubit maximally correlated states considered in 
\cite{Vidal-additivity}, where Wootters' formula \cite{Wootters} for 
calculating the entanglement of formation is used: Here we do not need the explicit 
formula of the coherence of formation, which involves an optimization problem itself, and 
indeed we do not know a closed-form expression for it in high dimension. However, here the 
equality constraint is so severe that we can learn the structure of the state. 

\begin{theorem}
\label{bound}
There is no bound coherence: $C_d(\rho)=0$ implies $C_c(\rho)=0$.
In other words, every state with any coherence 
(non-zero off-diagonal part) is distillable.
\end{theorem}

\medskip
{\it Discussion.}---
We have shown that the incoherent operations
proposed by Baumgratz \emph{et al.}~\cite{BCP14} give
rise to a well-behaved operational resource theory of coherence.
Remarkably, almost all basic questions in this resource theory have simple answers. 
We saw that it is a theory without bound coherence, but
exhibiting generic irreversibility for transformations between
pure and mixed states. This should be contrasted with the
general abstract framework of Brand\~{a}o and Gour~\cite{BrandaoGour},
which applies to the present theory of coherence, but rather
than the incoherent operations considered here, requires all
cptp maps $\mathcal{E}$ such that the weaker condition
$\mathcal{E}(\Delta) \subset \Delta$ holds. 
Both the distillable coherence and the coherence cost under this
relaxed premise become $C_r(\rho)$, meaning that while we cannot distill
more efficiently with this broader class of operations, 
formation becomes cheaper.

One curious observation, for which we would like to have a 
strictly operational interpretation, is that the resulting
theory of coherence resembles so closely the theory of \emph{maximally
correlated} entangled states. Indeed, under the correspondence
\[
  \rho = \sum_{ij} \rho_{ij} \ketbra{i}{j} \longleftrightarrow
  \widetilde{\rho} = \sum_{ij} \rho_{ij} \ketbra{ii}{jj},
\]
$C(\psi)$ is identified with $E(\widetilde{\psi})$, the entropic measure
of pure entanglement, $C_c(\rho)=C_f(\rho)$ is identified with
the entanglement cost (which equals the entanglement of formation
for these states), $E_c(\widetilde{\rho})=E_f(\widetilde{\rho})$,
and $C_d(\rho)=C_r(\rho)$ is identified with the distillable entanglement
(which equals the relative entropy of entanglement for these states),
$E_d(\widetilde{\rho})=E_r(\widetilde{\rho})$~\cite{Rains:PPT}. 
Indeed, this answers all the basic asymptotic questions in the theory, 
which is much simpler than general entanglement theory. 
What is missing to elevate
this correspondence from an observation to a theoretical
explanation (cf.~Streltsov \emph{et al.}~\cite{Adesso-et-al}) is a matching
correspondence between incoherent operations and LOCC operations.
That would truly show that the two theories are equivalent,
relieving us of the need to prove any of the optimal conversion rates 
reported in this Letter, which now we have to do by manually adapting
the entanglement manipulation protocols. 
A notable gap in the above correspondence is the non-asymptotic 
theory of pure states: The diagonal entries of a pure state
correspond to the Schmidt coefficients of the associated pure entangled state,
and by Nielsen's theorem~\cite{Nielsen:majorization} a pure state
can be transformed by LOCC into another one if and only if the
Schmidt vectors are in majorization relation. In the case of 
incoherent operations on pure states, we only know that majorization is
sufficient for transformability but not whether it is necessary.

To study this correspondence further, it might be worth investigating the
optimal conversion rates $R$ of incoherent transformations
\[
  \rho^{\ox n} \stackrel{\text{IC}}{\longmapsto} \,\,\stackrel{1-\epsilon}\approx \sigma^{\ox nR}
\]
for general (mixed) states, and which for the moment we can only
bound using the existing monotones:
\[
  \frac{C_r(\rho)}{C_f(\sigma)} \leq R 
          \leq \min \left\{ \frac{C_r(\rho)}{C_r(\sigma)}, \frac{C_f(\rho)}{C_f(\sigma)} \right\}.
\]

Given the close resemblance to entanglement theory, we may
expect that one-shot coding theorems as well as finite block length
analyses can be carried out, but we have refrained from entering this
domain to maintain the simplicity of the asymptotic picture. 
As in one-shot information theory~\cite{Tomamichel}, 
we expect that min- and max-entropies
and relative entropies and R\'{e}nyi (relative) entropies govern
the optimal rates, which now would carry an explicit dependence on
the protocol error.

It remains to be seen whether similarly complete theories of asymptotic
operational transformations can be carried out for other resource theories, 
such as that of reference frames~\cite{frame}. Observe that as reference 
frame theories are built on group actions under which the free states 
are precisely the invariant ones, the present theory of coherence may 
be viewed as the special case of the group of the diagonal phase unitaries.

\medskip
{\it Acknowledgments.}---
We thank Swapan Rana for spurring our interest in this project,
Gerardo Adesso, L\'{\i}dia del Rio, Jonathan Oppenheim and Alexander Streltsov for 
enlightening discussion on different ways of building a resource 
theory of coherence, and on resource theories in general,
and Sandu Popescu for a conversation on relations between the
resource theory of coherence and the measurement of time. Furthermore
Arun Pati and his co-authors for sharing a preliminary draft of their
paper~\cite{Pati-et-al}.

The present work was initiated when the authors were attending 
the programme ``Mathematical Challenges in Quantum Information'' (MQI) 
at the Issac Newton Institute in Cambridge, whose hospitality is 
gratefully acknowledged, 
and where DY was supported by a Microsoft Visiting Fellowship. 
The authors' work was supported by the European Commission (STREP ``RAQUEL''), 
the ERC (Advanced Grant ``IRQUAT''),
the Spanish MINECO (grant numbers FIS2008-01236 and FIS2013-40627-P)
with the support of FEDER funds, as well as by
the Generalitat de Catalunya CIRIT, project~2014-SGR-966.
DY furthermore acknowledges support from the NSFC, grant no.~11375165.

%
%
%

\appendix

\section*{Appendices}

\section{A. Miscellaneous facts and lemmas}
In this appendix, we collect standard facts about various functionals 
we use and show some lemmas that we need in the proofs of the main results.  

We use the Bures distance $B$ based on the fidelity $F$, but it is essentially 
equivalent to the trace norm:
\begin{align*}
  B(\rho,\sigma) &= \sqrt{2}\sqrt{1-F(\rho,\sigma)},\\
  F(\rho,\sigma) &= \tr\sqrt{\rho^{1/2}\sigma\rho^{1/2}}.
\end{align*}
Namely~\cite{FvdG},
\[
  \frac12 B(\rho,\sigma)^2 \leq \frac12 \|\rho-\sigma\|_1 \leq B(\rho,\sigma).
\]

In the proof of Theorem \ref{thm:distillable} below, 
we need the asymptotic continuity of $C_r$.

\begin{lemma}
The relative entropy of coherence is asymptotically continuous, 
i.e.~for $\|\rho-\sigma\|_1\leq\epsilon$,
\[
  |C_r(\rho)-C_r(\sigma)| \leq \epsilon\log d+2h(\epsilon/2), 
\]
where $d$ is the dimension of the supporting Hilbert space,
and $h(x)=-x\log x - (1-x)\log(1-x)$ is the binary entropy function.
\end{lemma}
\begin{proof}
The proof is straightforward because of Eq.~(\ref{eq:C_r}):
From $\|\rho-\sigma\|_1 \leq \epsilon$, we get 
$\|\Delta(\rho)-\Delta(\sigma)\|_1 \leq \epsilon$ 
by the monotonicity of the trace distance under the dephasing operation 
$\Delta$, which is a cptp map. 
By Audenaert's improvement~\cite{Audenaert} of Fannes' 
inequality~\cite{Fannes}, which states for two states 
$\rho$ and $\sigma$ on a $d$-dimensional Hilbert space such that
$\frac12 \|\rho-\sigma\|_1\leq\eta$, then 
$|S(\rho)-S(\sigma)|\le \eta\log d+h(\eta)$, we have
\begin{align*}
  |C_r(\rho)-C_r(\sigma)| &=    |S(\Delta(\rho))-S(\rho)-(S(\Delta(\sigma))-S(\sigma))|,\\
                          &\leq |S(\Delta(\rho))-S(\Delta(\sigma))|+|S(\rho)-S(\sigma)|,\\
                          &\leq \epsilon\log d+2h(\epsilon/2),
\end{align*}
concluding the proof.
\end{proof}

\medskip
In the proof of Theorem \ref{thm:cost} below, we need the monotonicity of the
coherence of formation,
\[
  C_f(\rho) = \min \sum_i p_i S(\Delta(\psi_i)) \text{ s.t. } \rho = \sum_i p_i \psi_i
\] 
under incoherent operations.

\begin{lemma}[{\AA{}berg~\cite{Aaberg:superposition}}]
  \label{lemma:CoF-monotone}
  The coherence of formation is convex and a strong coherence monotone: Indeed,
  for
  \[
    \rho \stackrel{{\rm IC}}{\longmapsto} \sigma = \sum_\ell K_\ell\rho K_\ell^\dagger 
                                                 = \sum_\ell p_\ell\sigma_\ell,
  \]
  it follows that
  \[
    C_f(\rho) \geq \sum_\ell p_\ell C_f(\sigma_\ell) \geq C_f(\sigma).
  \]
\end{lemma}
\begin{proof}
As a convex roof (aka convex hull), $C_f$ is automatically convex.
Furthermore, because of the convex roof property, we need to prove the strong 
monotonicity only for pure states. Consider a pure state $\psi$ and an incoherent 
operation $\mathcal{E}$ whose Kraus operator set is $\{K_\ell\}$, where each 
$K_\ell$ is incoherent operator, satisfying $\sum K_\ell^{\dagger}K_\ell=\1$. 
Suppose $K_\ell\ket{\psi} = \sqrt{p_\ell} \ket{\psi_\ell}$;
from this we construct a new incoherent operation $\tilde{\mathcal{E}}$ 
whose Kraus operator set is $\{\ket{\ell} \otimes K_\ell\}$ where the $\ket{\ell}$ 
are the basis states of an ancillary system. 
From the monotonicity of the relative entropy of coherence under the 
incoherent operation $\tilde{\mathcal{E}}$, we obtain
\[
  C_r(\psi) \geq C_r(\tilde{\mathcal{E}}(\psi)) 
            =    \sum_\ell p_\ell C_r(\psi_\ell)
            \geq C_f(\mathcal{E}(\psi)),
\]
and we are done.
\end{proof}

\medskip
Our proof for the asymptotic continuity of coherence of formation comes from 
that of entanglement of formation.

\begin{lemma}[{Nielsen~\cite{Nielsen-continuity}, Winter~\cite{AW-cont}}]
\label{lemma:EoF-cont}
For two bipartite states $\rho$ and $\sigma$ supported on a $d\times d$-dimensional Hilbert 
space with Bures distance $B(\rho,\sigma) \leq \epsilon \leq 1$, then
\[
\phantom{===:}
  |E_f(\rho)-E_f(\sigma)| \leq \epsilon\log d + (1+\epsilon) h\left(\!\frac{\epsilon}{1+\epsilon}\!\right).
\phantom{===:} \Box 
\]
\end{lemma}

\medskip
In the proof of Theorem \ref{thm:cost}, we need the asymptotic continuity of $C_f$.

\begin{lemma}
The coherence of formation is asymptotically continuous, 
i.e.~for $B(\rho,\sigma) \leq \epsilon \leq 1$,
\[
  |C_f(\rho)-C_f(\sigma)| \leq \epsilon\log d + (1+\epsilon) h\left(\!\frac{\epsilon}{1+\epsilon}\!\right),
\]
where $d$ is the dimension of the supporting Hilbert space.
\end{lemma}
\begin{proof}
Notice two facts: One is that 
\[\begin{split}
  \|\rho-\sigma\|_1
     &= \bigl\| \CNOT(\rho\otimes\proj{0})\CNOT^\dagger \bigr. \\
     &\phantom{===}
                \bigl. - \CNOT(\sigma\otimes\proj{0})\CNOT^\dagger \bigr\|_1,
\end{split}\]
the bipartite states on the right hand side being supported on a 
$d\times d$-dimensional Hilbert space, and the other is that 
$C_f(\rho)=E_f\bigl(\CNOT(\rho\otimes \proj{0})\CNOT^\dagger\bigr)$~\cite{Adesso-et-al}. 
From these and Lemma~\ref{lemma:EoF-cont} the claim follows.
\end{proof}

\medskip
In the proof of Theorem \ref{thm:distillable}, we need the following 
slight modification of~\cite[Prop.~2.4]{DW05}. 

\begin{lemma}[\protect{Cf.~Devetak/Winter~\cite[Prop.~2.4]{DW05}}]
  \label{lemma:covering}
  For a family $\{W_x\}_{x\in \cX}$ of quantum states on a $d$-dimensional
  Hilbert space $\cH$, and the type class $\cT^n_P=\{W^n_{x^n}\}$ w.r.t. an 
  empirical type $P$ of length-$n$-sequences over $\cX$, let $(Y_1,\ldots, Y_S)$ 
  be a random tuple whose elements are sampled from $\cT^n_P$ uniformly without
  replacement. Define the average state
  \[
    \sigma(P) = \frac{1}{|\cT^n_P|} \sum_{x^n\in\cT^n_P} W^n_{x^n}.
  \]
  Then, for every $0< \epsilon, \delta < 1$ and sufficiently large $n$,
  \[\begin{split}
    \Pr&\left\{ \left\| \frac1S \sum_{j=1}^S Y_j - \sigma(P) \right\|_1
               \geq \epsilon \right\}                                           \\
       &\phantom{=========}
        \leq 2 d^n \exp\left( -S\iota^n\frac{\epsilon^2}{288\ln 2}\right),
  \end{split}\]
  where $\iota = 2^{-I(X:W)+\delta}$ and $I(X:W)=S(X)+S(W)-S(XW)=S(\sum P_xW_x)-\sum P_xS(W_x)$ is the mutual information evaluated on the state $\omega=\sum_x P_x\ketbra{x}{x}\otimes W_x$.
\end{lemma}
\begin{proof}
Prop.~2.4 in \cite{DW05} follows directly from the matrix tail bound in~\cite{AW:covering} 
when the matrices are sampled uniformly with replacement, that is, $Y_j$ are i.i.d.~random 
variables. As pointed out in \cite{Gross}, the matrix tail bound in~\cite{AW:covering} 
still holds when the $Y_j$ are sampled uniformly without replacement \cite{GN}. 
So the modified \cite[Prop.~2.4]{DW05} also holds.
\end{proof}

\medskip
In the proof of Theorem \ref{thm:criterion}, we need the following lemma, 
which comes essentially from the structure of the state with 
vanishing conditional mutual information \cite{Hayden-equality}.
\begin{lemma}[{Winter/Yang~\cite{WY}}]
\label{lemma:WY}
A state $\rho_{AB}$ in the finite dimensional Hilbert space ${\cal H}_A\otimes {\cal H}_B$ 
satisfying
\[
  S(B)-S(AB)=E_f(AB)
\]
is of the form 
\[
  \rho^{AB}=\bigoplus p_i\rho_{i}^{B_i^L}\otimes\phi_{i}^{AB_i^R}, 
\]
where $\phi_{i}^{AB_i^R}$ are pure states and system $B$ is decomposed 
as a direct sum of tensor products:
${\cal H}_{B}=\bigoplus{\cal H}_{B_i^L}\otimes{\cal H}_{B_i^R}$.
\phantom{.}\hfill\qed
\end{lemma}

\section{B. Proofs}
Here we provide the detailed proofs of the results in the main text. 

\medskip
\begin{proofthm}{\bf of Theorem~\ref{thm:majorization}.} 
Assume $\operatorname{spec}\Delta(\psi) = \vec{p} 
\succ \vec{q} = \operatorname{spec}\Delta(\varphi)$,
and without loss of generality,
\[
  \ket{\psi} = \sum_i \sqrt{p_i}\ket{i},
  \quad
  \ket{\varphi} = \sum_i \sqrt{q_i}\ket{i},
\] 
since the diagonal unitaries required to adjust the phases
are incoherent operations.

It is well-known~\cite{MarshallOlkin} that
majorization implies that there is a probability
distribution $\lambda_\pi$ over permutations $\pi\in S_d$ such that
\[
  \vec{q} = \sum_\pi \lambda_\pi {\vec{p}}^\pi,
\]
where ${\vec{p}}^\pi$ is the vector $\vec{p}$ with indices
permuted according to $\pi$: ${\vec{p}}^\pi(i) = p_{\pi(i)}$.
Observe that we may without loss of generality assume the
$p_i$ and $q_i$ to be ordered non-increasingly, and also 
w.l.o.g.~that all $q_i > 0$, otherwise we reduce the Hilbert
space dimension $d$.

Now, define the Kraus operators
\[
  K_\pi := \sum_i \sqrt{\lambda_\pi}\sqrt{\frac{p_{\pi(i)}}{q_i}} \ketbra{\pi(i)}{i},
\]
which are evidently incoherent, and define a cptp map:
\[
  \sum_\pi K_\pi^\dagger K_\pi = \sum_{i\,\pi} \lambda_\pi \frac{p_{\pi(i)}}{q_i} \proj{i} = \1.
\]
Furthermore, it effects the desired transformation:
\[
 K_\pi\ket{\varphi} = \sum_i \sqrt{\lambda_\pi}\sqrt{p_{\pi(i)}} \ket{\pi(i)} 
                    = \sqrt{\lambda_\pi} \ket{\psi},
\]
for every permutation $\pi$.

Conversely, suppose $\varphi \stackrel{{\rm }}{\longmapsto} \psi$ by the incoherent operation with 
Kraus operators $\{K_\ell\}$, then $K_\ell\ket{\varphi} \propto \ket{\psi}$ 
for every $\ell$.  Since $\operatorname{rank} \Delta(\psi) = \operatorname{rank} \Delta(\varphi)$, 
we arrive at $K_\ell=\sum_i c_\ell(i)\ketbra{\pi_\ell(i)}{i}$, 
with a permutation $\pi_\ell$ on the basis states. 
The same form obtains from the assumption of a strictly incoherent map.
That is, $\sum_i c_\ell(i)\sqrt{q_i} \ket{\pi_\ell(i)} \propto \sum_i\sqrt{p_i} \ket{i}$. 
Now we show that
\[
  \ket{\varphi'}^{AB}\stackrel{\text{LOCC}}{\longmapsto}\ket{\psi'}^{AB},
\]
where 
\begin{align*}
  \ket{\varphi'}^{AB} &= \CNOT(\ket{\varphi}^A\ket{0}^B) = \sum_i \sqrt{q_i}\ket{i}^A\ket{i}^{B}, \\
  \ket{\psi'}^{AB}    &= \CNOT(\ket{\psi}^A\ket{0}^B)    = \sum_i \sqrt{p_i}\ket{i}^A\ket{i}^{B}.
\end{align*} 
Namely, Alice performs the measurement with Kraus operators $\{K_\ell\}$ and tells the 
outcome $\ell$ to Bob, who then performs the permutation transformation 
$\pi_\ell$ on system B. Thus they end up with the state
\[\begin{split}
  K_\ell^A\otimes \pi_\ell^{B}\ket{\varphi'}^{AB}
                     &=       \sum_i c_\ell(i)\sqrt{q_i}\ket{\pi_\ell(i)}^{A}\ket{\pi_\ell(i)}^{B} \\
                     &\propto \sum_i\sqrt{p_i}\ket{i}^{A}\ket{i}^B
                      =\ket{\psi'}^{AB}.
\end{split}\]
I.e., $\ket{\varphi'}^{AB}$ can be converted to $\ket{\psi'}^{AB}$ by an LOCC 
operation; from \cite{Nielsen:majorization}, we get that $\Delta(\psi) \succ 
\Delta(\varphi)$.
\end{proofthm}

\medskip
\begin{remark}
\normalfont
In \cite{DuBaiGuo}, it was claimed by a different approach that majorization
$\Delta(\psi) \succ \Delta(\varphi)$ follows already from a general 
incoherent transformation $\varphi \stackrel{{\rm IC}}{\mapsto} \psi$.
However, it seems that the proof has a gap, by way of some 
unjustified assumptions.
For the moment, Theorem~\ref{thm:majorization} is the best available statement.
\end{remark}

\medskip
\begin{proofthm}{\bf of Corollay \ref{cor:max-coherent-gen-all}.}
If $\rho$ is pure, then $\Phi_d \stackrel{\text{IC}}{\longmapsto} \rho$
follows from Theorem~\ref{thm:majorization}.

Any mixed state is a convex combination of pure states, $\psi_i$
with probability weight $p_i$, each of which can be obtained
from $\Phi_d$ by an incoherent map $\cE_i$. Then, 
$\cE = \sum_i p_i \cE_i$ is also incoherent and $\rho = \cE(\Phi_d)$. 
\end{proofthm}

\medskip
\begin{proofthm}{\bf of Theorem \ref{thm:asymptotic}.}
Without loss of generality, we suppose a general pure state in 
$d$-dimensional Hilbert space in the form 
$\ket{\psi}=\sum_{i=1}^d \sqrt{q_i} \ket{i}$ with 
a probability distribution $(q_i)$. 
The state of $n$ copies of $\ket{\psi}$ is
\[
  \ket{\psi}^{\ox n} = \sum_{i_1\ldots i_n} \sqrt{q_{i_1}\cdots q_{i_n} }\ket{i_1\ldots i_n}.
\]
We shall use information theory abbreviations for a generic
sequence of length $n$, $i^n = i_1\ldots i_n$, the probability
$q_{i^n} = q_{i_1}\cdots q_{i_n}$ and the state 
$\ket{i^n} = \ket{i_1\ldots i_n}$~\cite{CoverThomas}.
In analogy to entanglement concentration and dilution for bipartite pure state, 
we establish the conversion rate between the state $\ket{\psi}$ 
and the unit coherence state $\ket{\Phi_2}$. 

Concentration: Consider the coefficients and the basis sequences of $n$ copies 
of the state $\ket{\psi}$, we perform the type measurement $\{M_P, P\in {\cal P}_n\}$, 
where ${\cal P}_n$ is the set of \emph{types} of sequences with length $n$, and  
\[
  M_P=\sum_{i^n\in T(P)} \proj{i_1\ldots i_n},
\] 
with $T(P)$ the type class of $P$, i.e.~the set of sequences $i^n$ in which
the relative frequency of each letter $i$ equals $P(i)$.
By the law of large numbers, with probability converging to $1$, 
we get a type $P \approx Q$ (these are called \emph{typical}), which hence
has the property that $H(P)\approx H(Q)$, where $H(P) = -\sum_i P(i) \log P(i)$ is the
Shannon entropy of the distribution $P$.
Notice that this measurement is an incoherent operation, and that
the output state after application of $M_P$ is the maximally coherent state 
on the subspace spanned by the type class $T(P)$, which has dimensionality
\[
  2^{n H(P)}\geq |T(P)| \geq (n+1)^{-d} 2^{n H(P)}, 
\]
so by relabelling the indices of the state -- which can be realized by an 
incoherent unitary transformation --, the state can be transformed to 
$\ket{\Phi_2}^{\otimes nR}$ with $H(P)\geq R \geq H(P)-d\frac{\log(n+1)}{n}$, which 
tends to $H(P) \approx H(Q) = S(\Delta(\psi))$ when $n\to \infty$.  

Dilution: Consider the coefficients and the basis sequencers of n copies of
the state $\ket{\psi}$. We decompose the state into two terms 
\[
  \ket{\psi}^{\otimes n} = \sqrt{\Pr(\cT^n_{Q,\delta})} \ket{\text{typ}}
                            + \sqrt{1-\Pr(\cT^n_{Q,\delta})} \ket{\text{atyp}},
\]
with the typical part of the state, 
\[
  \ket{\text{typ}} = \frac{1}{\sqrt{\Pr(\cT^n_{Q,\delta})}}\sum \sqrt{q_{i_1}\cdots q_{i_n}}\ket{i_1\ldots i_n}, 
\]
summing over all $i_1\cdots i_n\in \cT^n_{Q,\delta}$, and an atypical rest $\ket{\text{atyp}}$.
Here, $\cT^n_{Q,\delta}$ the set of (entropy) typical sequences defined as 
\[
  \cT^n_{Q,\delta}:=\left\{i^n = i_1\ldots i_n: \left|-\frac{1}{n}\log (q_{i_1}\cdots q_{i_n})-H(Q) \right|
                                                                                          \leq \delta\right\}. 
\]
By the law of large numbers, for any $\epsilon >0$ and sufficiently large $n$, 
$\Pr(\cT^n_{Q,\delta})\geq 1-\epsilon$. 
Furthermore, $|\cT^n_{Q,\delta}|\leq 2^{n(H(Q)+\delta)}$. 
The dilution protocol uses then $n(H(Q)+\delta)$ copies of $\ket{\Phi_2}$, whose coefficients 
are uniform in dimension $2^{n(H(Q)+\delta)}$ and by Theorem \ref{thm:majorization} 
this state can be deterministically transformed by an incoherent operation
to any other pure state in dimension $2^{n(H(Q)+\delta)}$. 
So we can prepare the state $\ket{\text{typ}}$ satisfying 
$F(\proj{\text{typ}},\psi^{\otimes n}) \geq \sqrt{1-\epsilon}$, and 
the required rate of qubit maximally coherent states $\ket{\Phi_2}$ tends
to $H(Q) = S(\Delta(\psi))$. 

Optimality: We show both at the same time. Namely, consider a pure state $\ket{\psi}$
and asymptotic transformations 
$\psi^{\ox n} \stackrel{{\rm IC}}{\longmapsto} \approx \Phi_2^{nR-o(n)}$
and $\Phi_2^{\ox nR'+o(n)} \stackrel{{\rm IC}}{\longmapsto} \approx \psi^{n}$.
We know already that $R = C(\psi) = R'$ are achievable. Assume by contradiction
that one of the two can be improved, i.e.~$R > C(\psi)$ achievable for
concentration or $R' < C(\psi)$ is achievable for dilution. Thus,
we can have $\widetilde{R} = \frac{R}{R'} > 1$. 
Then, by composing the dilution and concentration protocols, we get an incoherent
transformation
\[
  \Phi_2^{\ox n} \stackrel{{\rm IC}}{\longmapsto} \approx \Phi_2^{n\widetilde{R}-o(n)}
\]
for large $n$. But by Lemma~\ref{lemma:rank-nonincrease}, each Kraus
operator $K_\ell$ of this transformation must produce a state $\ket{\psi_\ell}$
with $\operatorname{rank}\Delta(\psi_\ell) \leq 2^n$. From this one can readily derive
\[
  F^2(\psi_\ell,\Phi_2^{n\widetilde{R}-o(n)}) =    \tr\psi_\ell\Phi_2^{n\widetilde{R}-o(n)} 
                                              \leq 2^{-n(\widetilde{R}-R)-o(n)},
\]
and by convex combination, the output of any incoherent operation on
$\Phi_2^{\ox n}$ has exponentially small fidelity $\leq 2^{-n(\widetilde{R}-R)/2-o(n)}$ 
with $\Phi_2^{n\widetilde{R}-o(n)}$.
This contradiction shows that necessarily $R\leq C(\psi) \leq R'$.
\end{proofthm}

\medskip
\begin{proofthm}{\bf of Theorem \ref{thm:distillable}.}
The proof consists of two parts, the direct part showing that the claimed rate 
is achievable, and the converse part showing that no more can be distilled. 
In the direct part, we use the typicality technique and a covering lemma on 
operators. For the converse part, we follow the standard argument where we 
need the monotonicity, asymptotic continuity, and additivity of relative entropy 
of coherence.

Consider the purification
\[
  \ket{\phi}^{AE}=\sum_i \sqrt{p_i}\ket{i}^A\ket{\phi_i}^E
\]
of $\rho$, i.e.~$\tr_E\proj{\phi}^{AE}=\rho^{A}$. 
We perform the type measurement $\{M_P, P\in {\cal P}_n\}$, where ${\cal P}_n$ is the 
set of types of sequences with length $n$, and 
\[
  M_P=\sum_{i^n\in T(P)} \proj{i_1\ldots i_n},
\]
where $T(P)$ the type class of $P$ (see the proof of Theorem~\ref{thm:asymptotic}).

By the law of large numbers, with probability converging to $1$, we get as 
outcome a typical type $Q \approx P$, hence 
with the property $H(Q)\approx H(P)$. 
Then the bipartite state we are left with is 
\begin{equation}
  \label{eq:phi-Q}
  \ket{\phi(Q)}^{AE} = \frac{1}{\sqrt{|T(Q)|}} \sum_{i^n\in T(Q)} \ket{i^n}^A\ket{\phi_{i^n}}^E.
\end{equation} 
From Lemma \ref{lemma:covering}, we conclude that there exists a partition of 
the type class $T(Q)$ into $|T(Q)|/S=:M$ 
subsets $\{S_m\}$ with $|S_m|=S$ such that at least $M(1-\epsilon)$ of these
subsets of $\{S_m\}$ are ``good'' in the sense that the average states of system $E$ 
over the subsets is almost the same as a fixed state, the average state over $T(Q)$, 
\[
  \frac1S \sum_{i^n\in S_m} \phi_{i^n} \approx \frac{1}{|T(P)|} \sum_{i^n\in T(P)} \phi_{i^n}.
\]
The other, ``bad'', subsets are few.
Relabelling the indices $i^n \leftrightarrow (m,s)$, we can write 
Eq.~(\ref{eq:phi-Q}) as follows.
\[
  \ket{\phi(Q)}^{AE} 
     = \frac{1}{\sqrt{M}} \sum_{m=1}^M \ket{m} \ox \frac{1}{\sqrt{S}} \sum_{s=1}^S \ket{s}\ket{\phi_{ms}}.
\]
Now introduce
\[
  \ket{\phi(Q)_m} := \frac{1}{\sqrt{S}} \sum_{s=1}^S \ket{s}\ket{\phi_{ms}};
\] 
by Uhlmann's theorem \cite{Uhlmann,Jozsa}, for ``good'' $m$, there exists a unitary 
$U_m$ on $A$ such that $(U_m\otimes \1) \ket{\phi(Q)_m} \approx \ket{\phi(Q)_0}$. 
Now we perform the measurement $\{P_{\text{good}}, P_{\text{bad}}\}$,
and with probability $1-\epsilon$, we get the state 
\[
  \frac{1}{\sqrt{M(1-\epsilon)}} \sum_{\text{good }m} |m\>|\phi(Q)_m\>. 
\] 
Then we construct the incoherent operation whose Kraus operators are of the form
\[
  K_s=\sum_{\text{good }m} \proj{m}\otimes \ketbra{0}{s} U_{m}.
\]
This gives
\[
  K_s\ket{\phi(Q)}^{AE} 
     \approx \frac{1}{\sqrt{M(1-\epsilon)}} \sum_{\text{good }m} \ket{m}\ket{0} \bra{s} \phi(Q)_0\rangle,
\] 
which shows that we obtain an approximation to the coherent state 
$\frac{1}{\sqrt{M}} \sum_m \ket{m}$. 

It remains to estimate the quantity $M$ in the above protocol.
Observe $2^{n H(Q)}\ge|T(Q)|\geq (n+1)^{-d} 2^{n H(Q)}$ and $S=2^{nI(A:E)_\omega}$, 
where the mutual information is calculated with respect to the state 
$\omega = \sum_i q_i\proj{i}^A\otimes \proj{\phi_i}^E$. 
Since $Q$ is a typical type, we have
$H(Q)\approx H(P) = S(\Delta(\rho))$, and
\[
  I(A:E)_\omega = S(\omega^E) \approx S\left( \sum_i p_i \phi_i \right) = S(\rho^A).
\] 
Thus, we get $\frac{1}{n}\log M\to S(\Delta(\rho))-S(\rho)$ when $n\to \infty$.

Conversely, for any protocol $\cL_n$ such that 
$\left\| \cL_n(\rho^{\otimes n}) - \Phi_2^{\otimes nR} \right\|_1 \leq \epsilon$, 
we have
\begin{align}
nC_r(\rho) &=    C_r(\rho^{\otimes n}) \label{add1} \\
           &\geq C_r(\cL_n(\rho^{\otimes n})) \label{mono} \\
           &\geq C_r(\Phi_2^{\otimes nR}) - n\epsilon\log d-2h(\epsilon/2) \label{cont} \\
           &=    nR - n\epsilon\log d - 2h(\epsilon/2), \label{add2}
\end{align}
where Eqs.~(\ref{add1}) and (\ref{add2}) come from additivity and Ineq.~(\ref{mono})
is due to monotonicity and Ineq.~(\ref{cont}) due to asymptotic continuity.
So $R\le C_r(\rho) + \epsilon\log d + 2h(\epsilon/2)$.
When $\epsilon\to 0$ and $R \to C_d(\rho)$ as $n\to\infty$,
we obtain $C_d(\rho)\leq C_r(\rho)$.
\end{proofthm}

\medskip
\begin{proofthm}{\bf of Theorem \ref{thm:cost}.}
The proof consists of two parts, the direct part that shows we can prepare 
the state at the claimed rate, and the converse part that says that to prepare 
the state we have to consume $\Phi_2$ at least at that rate. 
The proof is very similar to the entanglement cost \cite{HHT-Ec}, except that 
here we have additivity which makes it simpler. In the direct part, we 
just prepare the typical part of the state and for the converse part, we 
still have the standard argument where we need the monotonicity, asymptotic 
continuity, and additivity of coherence of formation.

For an optimal convex decomposition, $\rho = \sum_i p_i \psi_i$, where the
indices $i$ range over an alphabet $\Omega$ (w.l.o.g.~of cardinality
$|\Omega| \leq d^2$, by Caratheodory's Theorem), we have 
\[
  \rho^{\ox n} = \sum_{i^n} p_{i^n} \psi_{i^n},
\]
with $i^n=i_1\ldots i_n$, $p_{i^n} = p_{i_1}\cdots p_{i_n}$
and $\psi_{i^n} = \psi_{i_1} \ox \cdots \ox \psi_{i_n}$.

The set of (frequency-)typical sequences,
\[
  \cT = \bigl\{i^n : \forall j\ |f_j(i^n)-p_j| \leq \delta_1 \bigr\},
\]
with $f_j(i^n) = \frac1n |\{t:i_t=j\}|$,
has the property that for large enough $n$,
\[
  \Pr(\cT) \geq 1-\epsilon_1.
\]

Now we rewrite $\rho^{\otimes n}=\rho(\cT)+\rho_0$, with the sub-normalized
state
\[
  \rho(\cT)=\sum_{i^n\in \cT}p_{i^n}\proj{\psi_{i^n}},
\]
and a rest $\rho_0$.
The protocol is now to sample an element $i^n\in\cT$ according to
the conditional distribution $\frac{1}{p^{\ox n}(\cT)}p^{\ox n}|_{\cT}$
of $p^{\ox n}$ restricted to ${\cal T}$. 
The number of occurrences of $j$ in each typical sequence $i^n \in \cT$ is
$N(j|i^n)\le n(p_j+\delta_1)$. 
By the conherence dilution protocol we can prepare a state
\[
  \rho_j^{(N(j|i^n))}\stackrel{1-\epsilon_2}{\approx} \psi_j^{\otimes N(j|i^n))}
\] 
by consuming at most $n(p_j+\delta_1)(S(\Delta(\psi_j))+\delta_2)$ copies of $\Phi_2$. 
So we can prepare
\[
  \bigotimes_j \rho_j^{(N(j|i^n))} =: \rho_{i^n} \stackrel{(1-\epsilon_2)^{|\Omega|}}{\approx} \psi_{i^n}, 
\]
using at most $\sum_j n(p_j+\delta_1)(S(\Delta(\psi_j))+\delta_2)$ copies of $\Phi_2$. 
Now we prepare a mixed state as the convex combination of these $\rho_{i^n}$, i.e.
\[
  \rho^{(n)} = \frac{1}{p^{\ox n}(\cT)} \sum_{i^n\in \cT}p_{i^n} \rho_{i^n}.
\]
By the joint concavity of the fidelity \cite{Uhlmann,Jozsa}, we get
\[\begin{split}
  F(\rho^{\otimes n},\rho^{(n)}) &\geq p^{\ox n}(\cT) F\left (\frac{\rho(\cT)}{p^{\ox n}(\cT)},\rho^{(n)} \right) \\
                                 &\geq (1-\epsilon_1)(1-\epsilon_2)^{|\Omega|}.
\end{split}\]
Since $|\Omega|$ is bounded, when $n\to \infty$, $\epsilon_1,\epsilon_2 \to 0$ and 
$F(\rho^{\otimes n},\rho^{(n)})\to 1$ for $\delta_1,\delta_2 \to 0$. 
The required rate of $\Phi_2$ tends to $ \sum p_iS(\Delta(\psi_i))=C_f(\rho)$. 
This shows that the rate $C_f(\rho)$ is asymptotically achievable.

For the optimality, consider an incoherent protocol that produces 
$\rho^{(n)}={\cal L}(\Phi_2^{\otimes m})$,
which is close to $\rho^{\ox n}$ up to error $\epsilon$, 
i.e.~$B\bigl(\rho^{(n)},\rho^{\ox n}\bigr) \leq \epsilon$. 
From the monotonicity, asymptotic continuity and additivity of $C_f(\rho)$, we get 
\begin{align}
  m &=    C_f(\Phi_2^{\otimes m}) \label{eq:1}   \\
    &\geq C_f(\rho^{(n)})         \label{ineq:2} \\
    &\geq C_f(\rho^{\otimes n}) - n\epsilon\log d
                                - (1+\epsilon) h\left(\!\frac{\epsilon}{1+\epsilon}\!\right) \label{ineq:3}\\
    &=    nC_f(\rho) - n\epsilon\log d
                     - (1+\epsilon) h\left(\!\frac{\epsilon}{1+\epsilon}\!\right),           \label{eq:4}
\end{align}
where Eqs.~(\ref{eq:1}) and (\ref{eq:4}) come from the additivity, 
Ineq.~(\ref{ineq:2}) from the monotonicity and Ineq.~(\ref{ineq:3}) from the 
asymptotic continuity.
Thus,
\[
  \frac{m}{n} \geq C_f(\rho)- \epsilon\log d - (1+\epsilon) h\left(\!\frac{\epsilon}{1+\epsilon}\!\right). 
\]
Letting $\epsilon \rightarrow 0$ as $n\rightarrow\infty$, 
we get $C_c(\rho)\geq C_f(\rho)$, as advertised.
\end{proofthm}

\medskip
\begin{proofthm}{\bf of Theorem \ref{additivity}.}
We observe that the coherence of formation $C_f(\rho)$ is equal to the 
entanglement of formation of the bipartite state $\CNOT(\rho\otimes \proj{0})\CNOT^\dagger$. 
The latter is a maximally correlated state in the sense that in any 
decomposition into pure states ensemble, the pure states have the same Schmidt basis. 
In \cite{Vidal-additivity}, the entanglement of formation for this class of states
is proved to be additive. So the coherence of formation is additive.

Additivity of $C_r$ comes directly from Eq.~(\ref{eq:C_r}) because 
$S(\rho\otimes\sigma)=S(\rho)+S(\sigma)$ and 
$S(\Delta(\rho\otimes\sigma))=S(\Delta(\rho))+S(\Delta(\sigma))$. 
\end{proofthm}

\medskip
\begin{proofthm}{\bf of Theorem \ref{thm:criterion}.}
Given a mixed state $\rho^A$, we construct the bipartite state 
$\sigma^{AB}=\CNOT(\rho^A\otimes \proj{0}^{B})\CNOT^\dagger$ which is a maximally
correlated state of the form $\sum_{ij}\rho_{ij}\ketbra{ii}{jj}$. 
Now the reversibility implies that $S(\sigma^B)-S(\sigma^{AB})=E_f(\sigma^{AB})$. 
Applying Lemma \ref{lemma:WY} to this and noticing that the $\sigma^A=\sigma^B$
we get
\[
  \sigma^{AB} = \bigoplus_j p_j \proj{\phi_j}^{AB},
\] 
where $\phi_j^B\perp\phi_k^B$ when $j\neq k$. Then using the $\CNOT$ again 
on $\sigma^{AB}$ we recover the original state $\rho^A$ of the form
\[
  \rho^A = \bigoplus_j p_j \proj{\phi_j}^{A},
\]
where each $\ket{\phi_j}^{A}$ is in the subspace ${\cal H}_j=\operatorname{span}\{ \ket{i} : i\in S_j\}$ 
and these subspaces are orthogonal to each other.
\end{proofthm}

\medskip
\begin{proofthm}{\bf of Theorem \ref{bound}.}
This follows from Theorem \ref{thm:distillable}, $C_d(\rho)=C_r(\rho)$ and the fact 
that $C_r(\rho)$ is a faithful coherence measure in the sense that $C_r(\rho)=0$ 
if and only if $\rho$ is incoherent. 
\end{proofthm}

\section{C. Generalized model}
While we restricted our treatment for simplicity of notation
to the case that the incoherent states are precisely a
fixed orthogonal basis of the Hilbert space and convex combinations,
we observe that our results carry over unchanged to the case
of a general decomposition into subspaces $\cH = \bigoplus_i \cH_i$.

In generalization of the picture in the main body of the present work, 
with an orthogonal basis and its superpositions, and inspired by 
{\AA}berg~\cite{Aaberg:superposition}, let us consider an
orthogonal decomposition of the Hilbert space $\cH$ into eigenspaces 
of some observable
\[
  O   = \sum_i E_i P_i,   \text{ so that }
  \cH = \bigoplus_i \cH_i,
\]
with an orthogonal direct sum, and $P_i$ is the projector
onto the subspace $\cH_i$ of $\cH$. With respect to this decomposition,
we declare states as \emph{incoherent} that respect the direct
sum decomposition:
\begin{equation}
  \label{eq:Delta}
  \Delta := \left\{ \rho = \sum_i q_i \rho_i : P_i \rho_i P_i = \rho_i \right\},
\end{equation}
where the $\rho_i$ are states (supported on $\cH_i$) and $(q_i)$
is a probability distribution. For the case of rank-one projectors
$P_i = \proj{i}$ we recover the notion of Baumgratz \emph{et al.}~\cite{BCP14}.

Slightly abusing notation, we also introduce the decohering projection map
\begin{equation}
  \label{eq:Delta-map}
  \Delta(A) = \sum_i P_i A P_i,
\end{equation}
such that $\Delta = \Delta\bigl(\cS(\cH)\bigr)$ is the image of the
state space under the decohering map; the image of the set of all
operators on ${\cal H}$ we denote by
\[
  \mathfrak{D} := \Delta\bigl(B(\cH)\bigr) = \bigoplus_i B(\cH_i).
\]

We call an operator $K$ acting on $\cH = \bigoplus_i \cH_i$,
or more generally mapping $\cH$ to $\cK = \bigoplus_j \cK_j$,
\emph{incoherent (IC)} if $K \cH_i \subset \cK_{j(i)}$ for every
$i$ and a function $i\mapsto j=j(i)$. If $j(i)$ is injective
-- which for operators mapping $\cH$ to itself means that
it is a permutation --, we call the operator 
\emph{strictly incoherent (strictly IC)}. The operator $K$ is strictly incoherent if
and only if $K$ as well 
as $K^\dagger$ are incoherent. When composing
systems, each with its own orthogonal decomposition,
$\cH = \bigoplus_i \cH_i$ and $\cH' = \bigoplus_j \cH_i'$,
we consider the tensor product space with the decomposition
$\cH \ox \cH' = \bigoplus_{ij} \cH_i \ox \cH_j'$. In this way,
the decohering map on the composite space becomes the tensor
product $\Delta \ox \Delta'$.

\medskip
With this, following~\cite{BCP14}, we can introduce the sets
of incoherent and strictly incoherent operations as special cptp maps:

\begin{align*}
  {\cal IC}   &:= \left\{ T \text{ cptp: } T(\rho) = \sum_\alpha K_\alpha \rho K_\alpha^\dagger,\ 
                                         \forall\alpha\ K_\alpha \text{ IC} \right\}, \\
  {\cal IC}_0 &:= \left\{ T \text{ cptp: } T(\rho) = \sum_\alpha K_\alpha \rho K_\alpha^\dagger,\ 
                                         \forall\alpha\ K_\alpha \text{ strictly IC} \right\}.
\end{align*}
By slight abuse of notation, we will write for a incoherent 
(strictly incoherent) Kraus operator $K$, that $K\in {\cal IC}$ ($K\in {\cal IC}_0$). 

Finally, the non-coherence-generating maps according to 
Brand\~{a}o and Gour~\cite{BrandaoGour} are
\[
  \widetilde{{\cal IC}} := \left\{ T \text{ cptp: } T(\Delta) \subset \Delta \right\},
\]
so that ${\cal IC}_0 \subsetneq {\cal IC} \subsetneq \widetilde{{\cal IC}}$.

\medskip
Our main results, Theorems~\ref{thm:majorization}, \ref{thm:asymptotic},
\ref{thm:distillable}, \ref{thm:cost}, and the additivity Theorem~\ref{additivity},
remain unchanged, as one can check. 

\medskip
\begin{remark}
  \label{rem:IC-IC0}
  \normalfont
  All of the transformations or asymptotic transformations
  contained in the above mentioned theorems can be effected
  by strictly incoherent operations (${\cal IC}_0$), except the distillation of mixed states.
  Indeed, it is not clear whether or not the rate $C_r(\rho)$
  is attainable with strictly incoherent operations. Our
  protocol for Theorem~\ref{thm:distillable}, at any rate, uses 
  ${\cal IC}$ in a non-trivial way.
\end{remark}

\end{document}